\documentclass[a4paper,final,conference]{IEEEtran}
\usepackage{amsthm}
\usepackage{amsmath}
\usepackage{pifont}
\usepackage[top=0.58in, bottom=0.58in, left=0.6in, right=0.6in]{geometry}

\usepackage{enumerate}
\usepackage{amsfonts}
\usepackage{graphicx}
\usepackage{latexsym}
\usepackage{amssymb}
\usepackage{stmaryrd}
\usepackage{amscd}
\usepackage[small]{caption}
\usepackage{array}
\usepackage{tikz,ifthen,calc}
\usetikzlibrary{trees,snakes,automata}
\usetikzlibrary{arrows,shapes,positioning}
\usepackage{comment}
\usepackage{multicol}
\usepackage{stackengine}

\newcommand{\cC}{{\cal C}}
\newcommand{\cD}{{\cal D}}

\newcommand{\cP}{{\cal P}}

\newcommand{\sG}{\script{G}}

\newcommand{\sP}{\script{P}}

\newcommand{\bfp}{{\boldsymbol p}}
\newcommand{\bfq}{{\boldsymbol q}}

\newcommand{\bfP}{{\mathbf P}}

\newcommand{\deff}{\mbox{$\stackrel{\rm def}{=}$}}

\newcommand{\ceilenv}[1]{\left\lceil #1 \right\rceil}

\newcommand{\sbinom}[2]{\left[ \begin{array}{c} #1 \\ #2 \end{array} \right] }

\DeclareMathAlphabet{\mathbfsl}{OT1}{cmr}{bx}{it}
\newcommand{\uuu}{\kern-1pt\mathbfsl{u}\kern-0.5pt}
\newcommand{\vvv}{\kern-1pt\mathbfsl{v}\kern-0.5pt}

\newcommand{\bC}{\overline{\cC}}
\newcommand{\balpha}{{\boldsymbol \alpha}}

\newcommand{\myboxplus}{\kern1pt\mbox{\small$\boxplus$}}

\makeatletter \DeclareRobustCommand{\sbinom}{\genfrac[]\z@{}}
\makeatother
\newcommand{\G}[2]{\sbinom{{#1}\kern-1pt}{{#2}\kern-1pt}}
\newcommand{\Gq}[2]{\sbinom{{#1}\kern-0.25pt}{{#2}\kern-0.25pt}}

\newcommand{\N}{\mathbb{N}}

\newcommand{\Ps}{\smash{{\sP\kern-2.0pt}_q\kern-0.5pt(n)}}
\newcommand{\sPs}{\smash{{\sP\kern-1.5pt}_q(n)}}
\newcommand{\Ptwo}{\smash{{\sP\kern-2.0pt}_2\kern-0.5pt(n)}}
\newcommand{\Ptwom}{\smash{{\sP\kern-2.0pt}_2\kern-0.5pt(m)}}
\newcommand{\Ptwonm}{\smash{{\sP\kern-2.0pt}_2\kern-0.5pt(n+m)}}
\newcommand{\Ptwoa}{\smash{{\sP\kern-2.0pt}_2\kern-0.5pt(1)}}
\newcommand{\Ptwob}{\smash{{\sP\kern-2.0pt}_2\kern-0.5pt(2)}}
\newcommand{\Ptwoc}{\smash{{\sP\kern-2.0pt}_2\kern-0.5pt(3)}}
\newcommand{\Ptwod}{\smash{{\sP\kern-2.0pt}_2\kern-0.5pt(4)}}
\newcommand{\Ptwoe}{\smash{{\sP\kern-2.0pt}_2\kern-0.5pt(5)}}
\newcommand{\Ptwof}{\smash{{\sP\kern-2.0pt}_2\kern-0.5pt(6)}}
\newcommand{\Ptwokm}{\smash{{\sP\kern-2.0pt}_2\kern-0.5pt(2k-1)}}
\newcommand{\Pone}{\smash{{\sP\kern-2.5pt}_2\kern-0.5pt(n{-}1)}}

\newcommand{\Gr}{\smash{{\sG\kern-1.5pt}_q\kern-0.5pt(n,k)}}
\newcommand{\Gi}{\smash{{\sG\kern-1.5pt}_q\kern-0.5pt(n,i)}}
\newcommand{\Gj}{\smash{{\sG\kern-1.5pt}_q\kern-0.5pt(n,j)}}
\newcommand{\Grmk}{\smash{{\sG\kern-1.5pt}_q\kern-0.5pt(n,n-k)}}
\newcommand{\Grdk}{\smash{{\sG\kern-1.5pt}_q\kern-0.5pt(2k,k)}}
\newcommand{\Grekappa}{\smash{{\sG\kern-1.5pt}_q\kern-0.5pt(n,e+1-\kappa)}}
\newcommand{\Grtwoekappa}{\smash{{\sG\kern-1.5pt}_q\kern-0.5pt(n,2e+1-\kappa)}}
\newcommand{\Gremkappa}{\smash{{\sG\kern-1.5pt}_q\kern-0.5pt(n,e-\kappa)}}
\newcommand{\Gn}{\smash{{\sG\kern-1.5pt}_2\kern-0.5pt(n,n{-}1)}}
\newcommand{\Gnq}{\smash{{\sG\kern-1.5pt}_q\kern-0.5pt(n,n{-}1)}}
\newcommand{\Gone}{\smash{{\sG\kern-1.5pt}_2\kern-0.5pt(n,1)}}
\newcommand{\Gqone}{\smash{{\sG\kern-1.5pt}_q\kern-0.5pt(n,1)}}
\newcommand{\GTwo}{\smash{{\sG\kern-1.5pt}_2\kern-0.5pt(n,k)}}
\newcommand{\GTwonk}[2]{{\smash{{\sG\kern-1.5pt}_2\kern-0.5pt({#1},{#2})}}}
\newcommand{\Gnk}{\smash{{\sG\kern-1.5pt}_2\kern-0.5pt(n,n{-}k)}}
\newcommand{\Greone}{\smash{{\sG\kern-1.5pt}_q\kern-0.5pt(n,e{+}1)}}
\newcommand{\Gretwo}{\smash{{\sG\kern-1.5pt}_q\kern-0.5pt(n,e{+}2)}}

\newcommand{\be}[1]{\begin{equation}\label{#1}}
\newcommand{\ee}{\end{equation}}

\newcommand{\Cref}[1]{Co\-rol\-la\-ry\,\ref{#1}}

\newcolumntype{L}[1]{>{\raggedright\let\newline\\\arraybackslash\hspace{-4pt}}p{#1}}
\newcolumntype{C}[1]{>{\let\newline\\\arraybackslash}p{#1}}

\newcommand{\tred}[1]{\textcolor{red}{#1}}
\newcommand{\tb}[1]{\textcolor{blue}{#1}}

\newtheorem{theorem}{Theorem}
\newtheorem{lemma}[theorem]{Lemma}
\newtheorem{remark}[theorem]{Remark}
\newtheorem{corollary}[theorem]{Corollary}

\newtheorem{problem}{Problem}
\newtheorem{problem*}{Problem}
\newtheorem{innercustomgeneric}{\customgenericname}
\providecommand{\customgenericname}{}
\newcommand{\newcustomtheorem}[2]{%
	\newenvironment{#1}[1]
	{%
		\renewcommand\customgenericname{#2}%
		\renewcommand\theinnercustomgeneric{##1}%
		\innercustomgeneric
	}
	{\endinnercustomgeneric}
}
\newcustomtheorem{customProblem}{Problem}
\newcustomtheorem{customlemma}{Lemma}
\newcustomtheorem{customthm}{Theorem}
\newcustomtheorem{customCor}{Corollary}

\begin{document}
	
	\title{\textbf{Iterative Programming of Noisy Memory Cells}
		\hspace{-0.5ex}
		}

\author{
\textbf{Michal Horovitz}\IEEEauthorrefmark{1},
\textbf{Eitan Yaakobi}\IEEEauthorrefmark{2},
\textbf{Eyal En Gad}\IEEEauthorrefmark{3},
and \textbf{Jehoshua Bruck}\IEEEauthorrefmark{4}\\[2mm]		 

\IEEEauthorblockA{\IEEEauthorrefmark{1} 
Dept. of Computer Science, Tel-Hai College,
and The Galilee Research Institute - Migal,
Upper Galilee, Israel
\\[1mm]}

\IEEEauthorblockA{\IEEEauthorrefmark{2}
Dept. of Computer Scienc, Technion---Israel Inst. of Tech.,
Haifa 32000, Israel  \\[1mm]}

\IEEEauthorblockA{\IEEEauthorrefmark{3} 
Micron Technology,
Milipitas, CA 95035, USA 
\\[1mm]}

\IEEEauthorblockA{\IEEEauthorrefmark{4} 
Dept. of Electrical Engineering,
California Inst. of Tech.,
Pasadena, CA 91125, USA
\\[1mm]}

{
Emails:\, 

\bfseries horovitzmic@telhai.ac.il,\, 
	 yaakobi@cs.technion.ac.il,\,  
	 eengad@micron.com,\,  
	 bruck@caltech.edu
}
\vspace{-4ex}} 
\vspace{-4ex}

	\maketitle
\begin{abstract}
In this paper, we study a model, which was first presented by Bunte and Lapidoth, that mimics the programming operation of memory cells. Under this paradigm we assume that cells are programmed sequentially and individually. The programming process is modeled as transmission over a channel, while it is possible to read the cell state in order to determine its programming success, and in case of programming failure, 
		to reprogram the cell again. Reprogramming a cell can reduce the bit error 
		rate, however this comes with the price of increasing the overall programming 
		time and thereby affecting the writing speed of the memory. An \emph{iterative programming 
			scheme} is an algorithm which specifies the number of attempts to program 
		each cell. Given the programming channel and constraints on the average and maximum number of attempts to program a 
		cell, we study programming schemes which maximize the number of bits that can be 
		reliably stored in the memory.  
		We extend the results by Bunte and Lapidoth and
		study this problem when the programming channel is either 
		the BSC, BEC, or $Z$ channel. For the BSC and the BEC our analysis is also extended for the case where the error probabilities on consecutive writes are not necessarily the same.
		Lastly, we also study a related model which is motivated by the synthesis process of DNA molecules.
	\end{abstract}\section{Introduction}\label{sec:intro}
\renewcommand{\baselinestretch}{0.98}\normalsize\noindent

	Many of existing and future volatile and non-volatile memories consist of memory cells. This includes for example DRAM, SRAM, phase-change memories (PCM), STT-MRAM, flash memories, as well as strands of DNA molecules. The information in these memories is stored in cells that can store one or multiple bits. The state of each cell can be changed in several methods which depend upon the memory technology, such as changing its resistance or voltage level. The process of changing the cell state, which we call here \emph{programming}, is crucial in the design of these memories as it determines the memory's characteristics such as speed, reliability, endurance, and more. Hence, optimizing the programming process has become an important feature in the development of these memories.

	
	Two of the more important goals when programming memory cells are speed and 
	reliability. In this work we aim to understand the relation between these two 
	figure of merits. Namely, we consider a model in which the cells are programmed 
	sequentially, one after the other \cite{BL10,BL11}. 
	Assume $n$ binary cells are programmed. The cell programming process is modeled 
	as transmission over some \emph{discrete memoryless channel} (\emph{DMC}) $C$, for example 
	the \emph{binary symmetric channel} (\emph{BSC}), the \emph{binary erasure channel} (\emph{BEC}), or the \emph{$Z$ channel}.
	It is assumed
	that when a cell is programmed
	we can check the success of its programming operation and in case of failure we 
	may choose to program it again. If there is no time restriction for programming 
	the cells, an optimal solution is to program each cell until it reaches its 
	correct value. For example, if the programming operation is modeled as the BSC 
	with crossover probability $p$, then the expected number of programming attempts until reaching success is $1/(1-p)$. 
	If $p=0.1$, this increases the programming operation time by roughly $11\%$. However, if the system allows to 
	increase the programming time by only $5\%$, then we search for a different 
	strategy. 
	
	More formally, we assume that there are~$n$ cells, for~$n$ sufficiently large, which are programmed according to some \emph{iterative programming scheme} $PS$. We define the \emph{average delay} of the programming scheme $PS$ over channel $C$ as the ratio between the expected number of programming attempts and the number of cells,~$n$, and the  \emph{maximum delay} is the maximum number of attempts to program a cell. Given some constraints, $D$ and $T$, on the average and maximum delay, respectively, our goal in this paper is to find a programming scheme that will maximize the number of information bits that can be reliably stored into these $n$ cells. Intuitively, the question is whether to spend time ensuring the cells are programmed correctly, or spend that resource for programming more redundancy cells in order to correct the errors. We present this problem as a solution to  program memory cells, however this is also a valid model for transmission processes where there is noise-free feedback on the transmission success. That is, we consider the problem of transmitting bits, or more generally packets, over a channel with feedback. Then, in case of transmission error, the goal is to determine an optimal strategy which specifies whether to retransmit the bit again. 
	
	Previous works considered programming schemes mostly for flash memory cells. 
	In~\cite{JiBrISITA08}, an optimal programming algorithm was presented to maximize the number of bits that can be stored in a single cell, which achieves the zero-error storage capacity under a noisy model. In~\cite{JiLiPACRIM09}, an algorithm was shown for optimizing the expected cell programming precision, when the programming noise follows a random distribution. In~\cite{YJSVW10}, algorithms for parallel programming of flash memory cells were studied which were then extended in~\cite{QYS14} as well as for the rank modulation scheme in~\cite{QJS13}.
	Other works studied the programming schemes with continuous-alphabet channels, see \cite{FLMS09,LFM10,LFMS08,LFMS14,LMF10,MFLES09,MLSF10,VTLF14} and references therein.
	
	Our point of departure in this paper is the programming model which was first presented in~\cite{BL10,BL11} by Bunte and Lapidoth for discrete alphabet memory  channels (DMC). In particular, in \cite{BL10} the case of symmetric channels with focus on the BSC was studied. We extend the results from \cite{BL10} for the BSC and study the problem for the BEC and the $Z$ channel, which the last is applicable in particular for flash memories. Furthermore, we also study the case when the error probabilities on consecutive programming operations are not the same. Even though we follow the model from~\cite{BL10}, we note that we propose a slightly different formulation to the problem and model, which we found to be more suitable to the cases we solve in this paper.
	
	Yet another model studied in this work is motivated by DNA-based storage systems. Recently, DNA has been explored as a possible near-future archival storage solution thanks to its potential high capacity and endurance~\cite{Betal16,B16,EZ17,Getal13,YGM16}. DNA synthesis is the process of artificially creating DNA molecules such that arbitrary single stranded DNA sequences of length few hundreds bases can be generated chemically. When synthesizing DNA strands, the bases are added one after the other to form the long sequence; for more details see~\cite{KC14}. However, this process is not prone to errors and several errors might occur in the form of insertions, deletions, and substitutions. Since the bases are added in a sequential manner it is possible to check the success of each step and thereby to correct failures or repeat the attachment of the bases. In particular, in case the attachment of a specific base does not succeed on several consecutive iterations, it is possible to add another different base which indicates a synthesis failure in this location.

	
	
	The rest of the paper is organized as follows. In Section~\ref{sec:def}, we formally present the definitions for the programming model and the problem studied in the paper. In Section~\ref{sec:mainProb}, we solve the programming model for the BSC and the BEC, and in Section~\ref{sec:Z-channel} we study the $Z$ channel. In Section~\ref{sec:generalProb1}, we generalize this problem for the setup where consecutive programmings of a cell do not necessarily behave the same with respect to the error probability. A new model motivated by DNA, which combines BSC and BEC is studied in Section~\ref{sec:psAsBSCandBEC}. Finally, Section~\ref{sec:conc} concludes the paper.

\section{Definitions and Basic Properties}\label{sec:def}
In this section we formally define the cell programming model and state the 
main problems studied in the paper. We also present some basic properties 
that will be useful in the rest of the paper.

Let $C$ be a discrete memoryless channel (DMC). We model the process of 
programming a cell as transmission over a channel $C$, with the distinction 
that after every programming attempt, it is possible to check the cell state 
and to decide, in the case of an error, whether to leave the cell erroneous, or 
reprogram it again. We assume that there are $n$ cells which are programmed 
individually. 
An \emph{iterative programming scheme}, or in abbreviation \emph{programing scheme}, is an algorithm which states the rules to program 
the $n$ cells.
Its \emph{average delay} over channel $C$
is defined to be the ratio between the expected 
number of programming attempts and the number of cells,
and the \emph{maximum delay} is the maximal number of attempts to program a cell.
Our primarily goal in 
this work is to reliably store a large number of bits into the cells, while 
constraining the average and the maximum delay.



We define a natural class of programming schemes which are denoted by 
$PS_t$, for $t\ge 0$, and $PS_{\infty}$. For $t\geq 0$, the strategy of the 
programming scheme $PS_t$ is to program the cell until its programming succeeds 
or the number of attempts is $t$, that is, after the $t$-th attempt the success 
is not verified and the cell may be left programmed errorneously. Applying $PS_0$ 
means 
that the cell is not programmed, while the programming scheme $PS_{\infty}$ is 
the one where the cell is programmed until it stores the correct value. For 
notational purposes in the paper, we denote the programming scheme $PS_{\infty}$ by 
$PS_{-1}$. 

For asymmetric channels the average delay of $PS_t$ may depend also on the code, for example in the $Z$ channel the average delay depends on the number of zeros in the codewords. Thus, from here on and until the end of this section we refer only for symmetric channels.
These concepts will be defined similarly in Section~\ref{sec:Z-channel} for the $Z$ channel.

For $t\geq -1$ and a symmetric channel, $C$, we denote by $D_t(C)$ the average delay of the programming scheme $PS_t$ when the programming process is modeled by the channel~$C$. For example (see Lemma~\ref{lem:errorProbP_D_R}), if the channel is the \emph{binary symmetric channel} (\emph{BSC}) with crossover probability $p$, 
or the \emph{binary erasure channel} (\emph{BEC}) with erasure probability~$p$, which will be denoted by $BSC(p)$ and $BEC(p)$, respectively, then $D_t(C(p)) = \sum_{i=0}^{t-1}p^i = 
\frac{1-p^t}{1-p}$ for $t\geq 0$ and $D_{-1}(C(p)) = 1/(1-p)$ (see Lemma~\ref{lem:errorProbP_D_R}),
where $C(p)$ is $BSC(p)$ or $BEC(p)$.
Unless stated otherwise, for the BEC we assume that $0\le p \le 1$ and for the BSC, $0\le p \le 0.5$.

When a cell is programmed according to a programming scheme $PS_t$, 
we can model this process as a transmission over $t$ copies of the channel $C$ 
and there is an error if and only if there is an error in each of the $t$ 
channels. 
We denote this as a new channel $C_{t}$. 
Note that a programming scheme has no 
effect on the types of the errors, 
but it may change the probability of the 
cell to be in error. For example, if one cell is programmed using the programming scheme $PS_{t_1}$ 
while another cell is programmed by the programming scheme $PS_{t_2}$, for $t_1\ne t_2$, then the probabilities of these 
cells to be erroneous may be unequal. 
We denote the capacity of the channel $C_{t}$ 
by $\cC_t(C)$. 
For example, for $t\ge 1$ and $C=BSC(p)$, $C_{t} = BSC(p^t)$ and the capacity of the channel $C_{t}$  is $\cC_t(C)=1-h(p^t)$\footnote{In this paper $h(x)$ is the binary entropy function where $0\le x \le 1$.}.
Note that for every channel, $C$, it holds that $\cC_0(C)=0$, $\cC_{-1}(C)=1$, and $D_0(C)=0$.

In this paper we focus on programming schemes that consist of combinations of 
several schemes from
$\{PS_t\}_{t\ge -1}$. 
Formally, given some $T\geq 0$, the maximum number of attempts to program a cell, we define the following set of 
programming schemes. 
\begin{align}\label{eq:programming schemes_T}
\cP_T\hspace{-0.5ex}=\hspace{-0.5ex}\bigg\{ & PS\left((\beta_1, t_1),(\beta_2, t_2),\ldots, 
(\beta_{\ell},t_{\ell})\right) :\\  \nonumber 
& 0\leq t_1, \ldots, t_\ell \leq T, 0<\beta_1,\ldots,\beta_\ell\leq 1,  \sum_{i=1}^{\ell}\beta_i=1 \bigg\},
\end{align}
where $PS\left((\beta_1, t_1),(\beta_2, t_2),\ldots, 
(\beta_{\ell},t_{\ell})\right)$
is a programming scheme of $n$ cells which works as follows. 
For all $1\leq i\leq \ell$, $\beta_i n$ of the cells are programmed according 
to  the programming scheme $PS_{t_i}$\footnote{We assume here and in the rest of the paper that $n$ is sufficiently large so that $\beta_i n$ is an integer number for all $i$.}. 
The set of programming schemes $\cP_{-1}$ is defined similarly where $-1\leq t_1, \ldots, t_\ell $.
\begin{align*}\label{eq:programming schemes}
\cP_{-1}\hspace{-0.5ex}=\hspace{-0.5ex}\bigg\{ & PS\left((\beta_1, t_1),(\beta_2, t_2),\ldots, 
(\beta_{\ell},t_{\ell})\right) :\\  \nonumber
& \hspace{-0.5ex} -1\leq t_1, \ldots, t_\ell, 0<\beta_1,\ldots,\beta_\ell\leq 1,  \sum_{i=1}^{\ell}\beta_i=1 \bigg\}.
\end{align*}

For $T\geq -1$, it can be readily verified that for a programming scheme $PS = 
PS\left((\beta_1, t_1),\ldots,(\beta_{\ell},t_{\ell})\right)\in\cP_T$ over a 
symmetric channel $C$, the average delay, denoted by $D_{PS}(C)$, is given by $$D_{PS}(C) =  
\sum_{i=1}^\ell\beta_iD_{t_i}(C).$$
Similarly, the \emph{capacity} of the programming scheme $PS$ over the channel $C$ is denoted by $\cC_{PS}(C)$ and is defined to be 
$$\cC_{PS}(C) = \sum_{i=1}^\ell\beta_i \cC_{t_i}(C),$$
where, as defined above, 
$\cC_{t_i}(C)$ is the capacity of the channel $C_{t_i}$.
Note that the definition of the capacity, $\cC_{PS}(C)$,  corresponds to the set 
of all achievable rates for reliably storing information in the cells. 
Specifically, when applying the programming scheme $PS$ to program cells over the channel $C$, the following properties hold:
\begin{itemize}
	\item for every $R<\cC_{PS}(C)$, there exists a sequence of codes 
	$\mathsf{C}_n=(2^{nR}, p_e^{(n)},n)$, such that $p_e^{(n)}\to 0$ as 
	$n\to\infty$,
	\item any sequence of codes $\mathsf{C}_n=(2^{nR}, p_e^{(n)},n)$ such that 
	$p_e^{(n)}\to 0$ as $n\to\infty$, must satisfy $R<\cC_{PS}(C)$,
\end{itemize}
where $\mathsf{C}_n$ is a code of size $2^{nR}$, $n$ is the length of the 
codewords, and 
$p_e^{(n)}$ is the decoding error probability when using the code 
$\mathsf{C}_n$. 

The main problem we study in this paper is formulated in 
Problem~\ref{prob:main} for  symmetric channels. 
The motivation of this problem is to maximize the number of information bits 
that can be reliably stored in $n$ cells when $n$ is sufficiently large,
where the average delay, that is, the average number of attempts to program a cell, is at most some prescribed value $D$,
and the number of attempts to program a cell is at most $T$, i.e., the maximum delay is at most $T$. The case of $T=-1$ corresponds to having no constraint on the maximum delay.
\begin{problem}\label{prob:main}
Given a symmetric channel $C$, an average delay $D$, and a maximum delay $T$,
find a programming scheme, $PS\in \cP_T$, which maximizes the capacity $\cC_{PS}(C)$, under the constraint that $D_{PS}(C)\le D$. In particular, given $C$ ,$D$, and $T$, find the value of $$F_1(C,D,T) = \max_{PS\in \cP_T:D_{PS}(C)\le D}\{\cC_{PS}(C)\}.$$
\end{problem}

Assume we are given a symmetric channel $C$, 
an average delay~$D$, 
and a programming scheme $PS=PS\left((\beta_1, t_1),\ldots, 
(\beta_{\ell},t_{\ell})\right)\in \cP_T$,
such that $D_{PS}(C) > D$.
In order to meet the constraint of the average delay $D$ by using the programming scheme $PS$,
we program only $\frac{D}{D_{PS}(C)}$ fraction of the cells
with the programming scheme $PS$, and the remaining cells are not programmed.
Hence, we define the programming scheme 
$PS(C,D)$ as follows:
\begin{equation}\label{eq:PS(C,D)}
\hspace{-0.1cm}	PS(C,D) \hspace{-0.15cm}=\hspace{-0.15cm}
	\begin{cases}
		\hspace{-0.1cm} PS, & \hspace{-1.28cm}\mbox{if } D_{PS}(C)\le D\\ 
		\hspace{-0.1cm} 
		PS\left((1\hspace{-0.1cm}-\hspace{-0.1cm}\beta,0),(\beta\beta_1, 
		t_1),\ldots, 
		(\beta\beta_{\ell},t_{\ell})\right)\hspace{-0.1cm},  & \hspace{-0.28cm} 
		\mbox{otherwise,} 
		\end{cases}
\end{equation}
where $\beta=\frac{D}{D_{PS}(C)}$.
It can be readily verified that the properties in the next lemma hold.
\begin{lemma}\label{lem:PS-D}
Given a symmetric channel $C$, 
an average delay $D$, 
and a programming scheme $PS\in \cP_T$, 
the following properties hold
\begin{enumerate}
\item  $D_{PS(C,D)}(C)=\min\{D_{PS}(C),D\}$, and
\item  $\cC_{PS(C,D)}(C)=\min\left\{1, \frac{D}{D_{PS}(C)}\right\} \cdot \cC_{PS}(C)$.
\end{enumerate}
\end{lemma}
Note that for $p>0$, $D_{PS}(C)=0$ if and only if $PS=PS((1,0))$, and then we define $\cC_{PS(C,D)}(C)=\cC_{PS}(C)$ which is equal to zero by the definition of $PS_0$.

We next state another concept 
which will be helpful in solving Problem~\ref{prob:main}.
The \emph{normalized capacity} of
 a programming scheme $PS$ over a symmetric channel $C$ is defined to be 
\begin{equation}\label{eq:normalizedCapacity}
{\bC}_{PS}(C) =\begin{cases}
 \frac{\cC_{PS}(C)}{D_{PS}(C)}, & \mbox{if } D_{PS}(C)>0,\\
 \cC_{PS}(C), & \mbox{otherwise.} 
\end{cases}
\end{equation}
The normalized capacity is the ratio between 
the maximum number of information bits that can be reliably stored and 
the average number of programming attempts.

Lemma~\ref{lem:connect_prob1AndProb2_1}
presents a strong connection between 
the normalized capacity of a programming scheme $PS$ over a channel $C$
and its capacity over a channel $C$ under a constraint $D$.
\begin{lemma}\label{lem:connect_prob1AndProb2_1}
	For a symmetric channel $C$, 
	an average delay $D$, 
	and a programming scheme $PS$,
	the following holds
	$$
	\cC_{PS(C,D)}(C)=\min\{D,D_{PS}(C)\}\cdot {\bC}_{PS}(C)
	.
	$$
\end{lemma}
\begin{proof}
	By Lemma~\ref{lem:PS-D}, 
	if $D_{PS}(C)\ge D$ then
	\begin{align*}
	\cC_{PS(C,D)}(C)
	& =\frac{D}{D_{PS}(C)}\cdot\cC_{PS}(C)\\
	& =D\cdot\frac{\cC_{PS}(C)}{D_{PS}(C)}=D\cdot{\bC}_{PS}(C).
	\end{align*}
	Otherwise, $D_{PS}(C)<D$ 
	and $\cC_{PS(C,D)}(C)=\cC_{PS}(C)$ 
	by the definition of $PS(C,D)$, and 
	$\cC_{PS}(C)=D_{PS}(C)\cdot {\bC}_{PS}(C)$,
	by the definition of the normalized capacity.
\end{proof}

In this paper we study the BSC, the BEC, and the $Z$ channel. In these channels, the cells store binary information, where in the BSC a programming failure changes the bit value in the cell, 
in the BEC, a failure causes an erasure of an information bit, and lastly in the $Z$ channel only the programming of cells which are programmed with value zero can fail.

For the $Z$ channel, which is not a symmetric channel, the average delay depends also on the code, in particular, on the number of zeros in the codewords.
Thus, the $Z$ channel is discussed in a different section,
Section~\ref{sec:Z-channel},  in which we state similar definitions to Problem~\ref{prob:main} and to the related concepts, $D_t(C)$,  $PS(C,D)$, and the normalized capacity.

The following table summarizes most of the notations used in this paper.
\begin{table}[h]
	\begin{center}
		\caption{Summary of notations. \label{tab:abbrev}}
		\begin{tabular}{| C{1.4cm}  | C{6.6cm} |}
			\hline
			Notation &  Description 
			\\
			\hline
			$PS_t$ & Programming scheme with at most $t$ attempts 
			\\
			\hline
			$\cP_T$ & The set of all programming schemes with maximum delay $T$  
			\\
			\hline
			$D_{PS}(C)$&  
			The average delay of $PS$ over the channel $C$ 
			\\
			$D_{t}(C)$ &  
			The average delay of $PS_t$ over the channel $C$ 
			\\
			$D_t(p)$ &  
			The average delay of $PS_t$ over the $BSC(p)$ or $BEC(p)$ 
			\\
			\hline
			$\cC_{PS}(C)$&  
			The capacity of $PS$ over channel $C$ 
			\\
			$\cC_{t}(C)$ &  
			The capacity of $PS_t$ over channel $C$ 
			\\
			\hline
			$F_1(C,D,T)$  Problem~\ref{prob:main} &
			The maximum capacity of channel $C$ using $PS\in \cP_T$ under an average delay constraint $D$ 
			\\
			\hline
			$PS(C,D)$ & Adjusted $PS$ to meet the constraint $D$ for channel $C$ 
			\\	
			\hline
			$\bC_{PS}(C)$&  
			The normalized capacity of $PS$ over channel $C$ 
			\\
			$\bC_{t}(C)$ &  
			The normalized capacity of $PS_t$ over 	channel $C$ 
			\\
			\hline
			$BSC(p)$ &  
			The binary symmetric channel with crossover probability~$p$, $0\le p\le 0.5$ 
			\\
			$BEC(p)$ &  
			The binary symmetric channel with erasure probability~$p$, $0\le p\le 1$
			\\
			$Z(p)$ &  
			The Z channel  with error probability $p$, $0\le p\le 1$
			\\
			\hline
			$Z(p,\alpha)$ &  
			The Z channel  with error probability $p$ and $\alpha$ ones in each codeword, $0\le p, \alpha\le 1$
			\\
			\hline
			$PS_{q,t}$ & $PS_t$ where in the last attempt a question-mark is written with probability $1-q$ 
			\\
			\hline
			$\cC_{q,t}(C)$ & The capacity of $PS_{q,t}$ over channel $C$  
			\\
			\hline
		\end{tabular}
	\end{center}
\end{table}

\section{
	The BSC and the BEC}
\label{sec:mainProb}
In this section we study Problem~\ref{prob:main} for the BSC and the BEC.
Note that the results for the BSC have already been studied in \cite{BL10}, however we present them here in order to compare with the BEC and since these results will be used in Section~\ref{sec:generalProb1} for the case of programming with different error probabilities, and in Section~\ref{sec:psAsBSCandBEC} for a new model. 
Additionally, the translation between the notations by Bunte and Lapidoth~\cite{BL10} and our formulation, is not immediate, and hence we found this repetition to be important for the readability and completeness of the results in the paper.
For the same reasons, we provide proofs for some of the results on the BSC in this section.

 According to well known results on the capacity of the BSC and the BEC we 
 first establish the following lemma.
\begin{lemma}\label{lem:errorProbP_D_R}
For the programming scheme $PS_t$, $t\ge -1$, and error probability $p$ for 
the BSC and the BEC, the following properties hold:
\begin{enumerate}
\item For all $t\ge 1$, $\cC_{t}\left(BSC(p)\right)=1-h(p^t)$,
\item For all $t\ge 1$, $ \cC_{t}\left(BEC(p)\right)=1-p^t$,
\item $\cC_{-1}\left(BSC(p)\right)=\cC_{-1}\left(BEC(p)\right)=1$,
\item For all $t \ge 0$,\vspace{-2ex} $$D_{t}(p)\deff D_{t}\left(BSC(p)\right)=D_{t}\left(BEC(p)\right)=\frac{1-p^t}{1-p},$$
\vspace*{-2ex}
\item $D_{-1}(p)\deff D_{-1}\left(BSC(p)\right)=D_{-1}\left(BEC(p)\right)= \frac{1}{1-p}$.
\end{enumerate}
\end{lemma}
\begin{proof}

For the programming scheme $PS_t$, $t\ge 1$, a cell will be erroneous if all its~$t$ programmings have failed, which happens with probability~$p^t$. According to the known results on the capacity of the BSC and the BEC, we conclude claims 1, 2 and 3 in the lemma regarding the capacity of $PS_t$ over the channels $BSC(p)$ and $BEC(p)$.

The average delay of $PS_t$ for $t\ge 1$ is computed as follows. Let $q_i$ be the probability that a cell is programmed at least $i$ times, $i\ge 1$. Thus, $q_i=p^{i-1}$, and the average number of attempts to program a cell both for the BSC and the BEC is equal to $\sum_{i=1}^t q_i$. Then, we conclude that \vspace{-1ex}$$D_t(p)=\sum_{i=1}^t q_i =\sum_{i=0}^{t-1}p^i = \frac{1-p^t}{1-p}.$$
For $t=0$ the average delay is zero, and for $t=-1$ the average delay is 
\vspace*{-1ex}
$$D_{-1}(p)=\sum_{i\ge 1} q_i =\sum_{i\ge 0}p^i = \frac{1}{1-p}.\vspace{-2ex}
\vspace{-2ex}$$
\end{proof}

The next theorem compares between the normalized capacity of $PS_t$ and $PS_{t+1}$ over the $BSC(p)$ and the $BEC(p)$, for each $t\ge 1$. This result is used next in Corollary~\ref{cor:solutionProb1} which establishes the solution to Problem~\ref{prob:main} for these two channels. 
\begin{theorem}\label{thm:errorProbP_NORM-R}
For all $t\geq 1$ the following properties hold:
\begin{enumerate}
\item ${\bC}_{PS_{t}}\left(BSC(p)\right) \leq {\bC}_{PS_{t+1}}\left(BSC(p)\right)$,
\item ${\bC}_{PS_{t}}\left(BEC(p)\right)=1-p$.
\end{enumerate}
\end{theorem}
\begin{proof}
It is possible to verify that the function $$f(x) = \frac{1-h(x)}{1-x}$$ is decreasing in the range $0\le x \le 0.5$,
and by Lemma~\ref{lem:errorProbP_D_R} we get
$${\bC}_{PS_{t}}\left(BSC(p)\right) = \frac{\cC_{t}\left(BSC(p)\right)}{D_{t}\left(BSC(p)\right)} = (1-p)\cdot f(p^t).$$
Thus, 
\begin{align*}
{\bC}_{PS_{t}}\left(BSC(p)\right) 
& = (1-p)\cdot f(p^t) \\
& \leq  (1-p)\cdot f(p^{t+1})={\bC}_{PS_{t+1}}\left(BSC(p)\right).
\end{align*}
For the BEC, by Lemma~\ref{lem:errorProbP_D_R}
we have $${\bC}_{PS_{t}}\left(BEC(p)\right)=\frac{\cC_{t}\left(BEC(p)\right)}{D_{t}\left(BEC(p)\right)} =1-p.\vspace{-2ex}$$
\end{proof}

The solutions to Problem~\ref{prob:main}
for the $BSC(p)$ and the $BEC(p)$ are presented in
Corollary~\ref{cor:solutionProb1},
where the result for the BSC has already presented in \cite{BL10}; see Proposition~5 therein.
Note that if $D\geq\frac{1}{1-p}$
then the average delay is not constrained, since the average delay of any $PS$ does not exceed the average delay of $PS_{-1}$ which equals to $\frac{1}{1-p}$
(see also~\cite{BL10}). 
\begin{corollary}\label{cor:solutionProb1}
For $T\geq -1$, denote $D'=\min\{ D_T(p),D\}$.
The solution to Problem~\ref{prob:main} for the $BSC$ and the $BEC$ is as follows.
\begin{enumerate}
\item If $T\ge 0$ then
\begin{enumerate}
	\item $F_1(BSC(p), D,T) =D' \cdot 	\frac{(1-p)(1-h(p^T))}{1-p^T}$ and this value is obtained by the programming scheme $PS_T(BSC(p),D)$
\item $F_1(BEC(p), D,T) = D'\cdot(1-p)$ and this value is obtained by the programming scheme  $PS_t(BEC(p),D)$ 
for any $t$ such that $0\le t \le T$ and $D_{t}(C)\ge D'$,
\end{enumerate}
\item $F_1(BSC(p), D,-1) = F_1(BEC(p), D,-1) =D'\cdot (1-p)$ and this value is obtained by the programming scheme $PS_{-1}(BEC(p),D)$ for the $BSC$, and by the programming scheme $PS_{t}(BEC(p),D)$ for any $t$ such that $D_{t}(C)\ge D'$ for the $BEC$.
\end{enumerate}
\end{corollary}
\begin{proof}
In order to find the value of $F_1(C,D,T)$ where $C$ is either the $BSC(p)$ or the $BEC(p)$, we let $PS = PS\left((\beta_1, t_1),\ldots, (\beta_{\ell},t_{\ell})\right) \in \cP_T$ be a programming scheme which meets the constraint $D$, that is,
$$D_{PS}(C)=\sum_{i=1}^{\ell}  \beta_i\cdot D_{t_i}(C)\le \min \{D, D_T(C)\}.$$
Then, for $C=BSC(p)$ the capacity of the programming scheme $PS$ over $C$ 
satisfies\vspace{-2ex}
\begin{align*}
\cC_{PS}(C)
 & =\sum_{i=1}^{\ell} 
 \beta_i\cdot \cC_{t_i}(C)\\
 & \overset{(1)}{=}\sum_{i=1}^{\ell} 
 \beta_i\cdot D_{t_i}(C) \cdot {\bC}_{PS_{t_i}}(C)\\
 & \overset{(2)}{\le}\sum_{i=1}^{\ell} 
 \beta_i\cdot D_{t_i}(C) \cdot {\bC}_{PS_{T}}(C)\\
 & = {\bC}_{PS_{T}}(C) \sum_{i=1}^{\ell} 
 \beta_i\cdot D_{t_i}(C)\\
 & \overset{(3)}{\le} {\bC}_{PS_{T}}(C) \cdot \min \{D, D_T(C)\} \\
 & \overset{(4)}{=} 
 \cC_{PS_{T}(C,D)}(C),
\end{align*}
where $(1)$ is by the definition of the normalized capacity,
$(2)$ is by Theorem~\ref{thm:errorProbP_NORM-R},
$(3)$ is by 
$D_{PS}(C)=\sum_{i=1}^{\ell} 
\beta_i\cdot D_{t_i}(C)\le \min \{D, D_T(C)\}$,
and $(4)$ is by 
Lemma~\ref{lem:connect_prob1AndProb2_1}.
A similar proof holds for $F_1(BEC(p),D,T)$.
\end{proof}

\begin{remark}
The claims in Lemma~\ref{lem:errorProbP_D_R} regarding the BSC were presented in Proposition~3 in~\cite{BL10}, and the result in Corollary~\ref{cor:solutionProb1} for the BSC was presented in Proposition~5 in~\cite{BL10}. We note that in Proposition~3 in~\cite{BL10}, $\epsilon, \zeta$ is equivalent to $p, D-1$ in our notations, respectively. Furthermore, the gap in the solution from~\cite{BL10} and our result stems from the fact that we let cells to be not programmed at all, while in \cite{BL10} a cell has to be programmed at least once. Thus, the translation between these two approaches can be done by substituting the average delay constraint $D$ with $\zeta+1$.
\end{remark}

\section{The $Z$ Channel}
\label{sec:Z-channel}
In this section we study programming schemes for the $Z$ channel with error probability $p$, i.e., $0$ is flipped to $1$ with probability $p$, where $0\le p\le1$. This channel is denoted by $Z(p)$.

The capacity of the channel $Z(p)$ was well studied in the literature; see e.g.~\cite{TAB02,V97}. 
We denote by $Z(p,\alpha)$ the $Z$ channel where $\alpha$ is the probability for occurrence of $1$ in a codeword, and $p$ is the crossover $0\to 1$ probability. The capacity of $Z(p,\alpha)$ was shown to be~\cite{TAB02,V97}
$$\mathsf{cap}(Z(p,\alpha)) \deff h((1-\alpha)(1-p)) - (1-\alpha)h(p). $$

In the $Z$ channel, 
the average delay of programming a zero cell is exactly as in the BSC and the BEC cases, 
but a cell with one value is programmed only once. 
Therefore, the average delay depends on the number of cells which are programmed with zero, and hence we define $D_t(Z(p,\alpha))$ as the average delay of the programming scheme $PS_t$ when the programming process is modeled by the channel $Z(p)$ and $\alpha$ is the fraction of ones in the codewords.
The capacity $\cC_t(Z(p,\alpha))$ is defined to be the capacity of the channel $Z(p)$ when $\alpha$ is the probability for occurrence of one in the codewords and the programming scheme $PS_t$ is applied.
The following lemma is readily proved.
\begin{lemma}\label{lem:errorProbP_Zchannel}
	For the programming scheme $PS_t$ and the channel $Z(p,\alpha)$, the following properties hold:
	\begin{enumerate}
	\item For all $t\ge 0$, 
	\begin{align*}
	\cC_{t}\left(Z(p,\alpha)\right)
	& = 
	\mathsf{cap}(Z(p^t,\alpha)) \\
	& =h((1-\alpha)(1-p^t)) - (1-\alpha)h(p^t),
	\end{align*}
	\item 
	$\cC_{-1}\left(Z(p,\alpha)\right)=h(\alpha)$,
	\item For all $t \ge 1$, 
	$D_t(Z(p,\alpha))=
	\frac{(1-\alpha)(1-p^t)}{1-p} + \alpha$,
	\item $D_{-1}(Z(p,\alpha))
	= \frac{1-\alpha}{1-p}+\alpha$, $D_0(Z(p,\alpha))=0$.
\end{enumerate}
\end{lemma}


Let $PS\left((\beta_1, t_1),(\beta_2, t_2),\ldots, (\beta_{\ell},t_{\ell})\right)$ be a programming scheme, and
$\balpha=(\alpha_1,\alpha_2,\ldots,\alpha_{\ell})$
where $0\le \alpha_i \le 1$, for all $1\leq i\leq \ell$.
Then, we define 
$$\cC_{PS}(Z(p,\balpha))=\sum_{i=1}^{\ell} 
\beta_i\cdot \cC_{t_i}(Z(p,\alpha_i)),$$
that is $\cC_{PS}(Z(p,\balpha))$ is the capacity of $Z(p)$ while using the programming scheme $PS$ and the parameter $\balpha$.
Similarly, we define the average delay of the programming scheme $PS$ for $Z(p)$ using the parameter $\balpha$ as
$$D_{PS}(Z(p,\balpha))=
\sum_{i=1}^{\ell} 
\beta_i\cdot \cD_{t_i}(Z(p,\alpha_i)).$$

Thus, we formulate 
Problem~\ref{prob:main} for the $Z$ channel as follows.
\begin{customProblem}{\ref{prob:main} - $Z$ channel}\label{prob:main_Z}
	Given the channel $Z(p)$, an average delay $D$, and a maximum delay $T$, find a programming scheme, $PS\in \cP_T$, and a vector $\balpha$ which maximize the capacity $\cC_{PS}(Z(p,\balpha))$, under the constraint that $D_{PS}(Z(p,\balpha)))\le D$. 
	In particular, given $Z(p)$, $D$, and $T$, find the value of
	$$F_1(Z(p),D,T) = \max_{PS\in \cP_T,\balpha \ :D_{PS}(Z(p,\balpha))\le D}\{\cC_{PS}(Z(p,\balpha))\}.$$
\end{customProblem}


In order to solve Problem~\ref{prob:main} for the $Z$ channel, we use the normalized capacity of a programming scheme 
$PS_t$ over $Z(p,\alpha)$ which is defined as in Equation~(\ref{eq:normalizedCapacity}) for $t\ne0$ by
$$
	{\bC}_{t}(Z(p,\alpha)) = \frac{\cC_{t}(Z(p,\alpha))}{D_{t}(Z(p,\alpha))},
$$
and ${\bC}_{0}(Z(p,\alpha))=0$.
For $t\ge 0$ we have
$$
{\bC}_{t}(Z(p,\alpha))=\frac{(1-p)(h((1-\alpha)(1-p^t)) - (1-\alpha)h(p^t))}{(1-\alpha)(1-p^t)+\alpha(1-p)},
$$
and for $t=-1$ it holds that 
$${\bC}_{-1}(Z(p,\alpha))=\frac{(1-p)h(\alpha)}{(1-p\alpha)}.$$

Given $p,t$, the maximum normalized capacity of $PS_t$
is ${\bC}_{t}(Z(p)) = \max_{0\le \alpha \le 1}\left\{{\bC}_{t}(Z(p,\alpha))\right\}$,
and we denote by ${\alpha}^*(p,t)$ the value of $\alpha$ which achieves this capacity.
That is, $${\bC}_{t}(Z(p))={\bC}_{t}(Z(p,{\alpha}^*(p,t)))=\max_{0\le \alpha \le 1}\left\{{\bC}_{t}(Z(p,\alpha))\right\},$$
and the average delay $D_{t}(Z(p))$ is defined by
$$D_{t}(Z(p))=D_{t}(Z(p,\alpha^*(p,t))).$$

Next, we define the programming scheme $PS_t(Z(p),D)$ similarly to the definition in Equation~(\ref{eq:PS(C,D)}).
$PS_t(Z(p),D)$ is a scheme in which the cells are programmed by $PS_t$ 
until the average delay is $D$, and then the rest of the cells are not programmed.
That is, denote by $\beta=\frac{D}{D_{t}(Z(p))}$, and
$$
PS_t(Z(p),D) =
\begin{cases}
PS_t, &\hspace{-1cm} \mbox{if } D_{t}(Z(p))\le D\\ 
PS\left((1-\beta,0),(\beta, t)\right) ,  & 
\mbox{otherwise.} 
\end{cases}
$$
Given $T$, a constraint on maximum delay, we define  $t^*(T)=\arg \max_{0 \le t\le T} \{{\bC}_{t}(Z(p))\}$ for $T\ge0$ 
and $t^*(-1)=\arg \max_{-1 \le t} \{{\bC}_{t}(Z(p))\}$.

Thus, we can conclude the following corollary which is proved in a similar technique of 
Corollary~\ref{cor:solutionProb1}. The proof is presented in the appendix.
\begin{corollary}\label{cor:solutionProb1Z-channel}
	$F_1(Z(p), D,T) =\min\{D_T(Z(p)),D\}\cdot{\bC}_{PS_{t^*(D)}}\left(Z(p)\right)=
	{\cC}_{PS_{t^*(D)}(Z(p),D)}\left(Z(p)\right)$
	and this value is obtained by $PS_{t^*(T)}(Z(p),D)$ with parameter $\alpha^*(p,t^*(T))$.
\end{corollary}

An explicit solution for the $Z$ channel can be obtained by finding the value of $t^*(T)$ and $\alpha^*(p,t^*(T))$.
We could not solve it explicitly, however we present some computational results.
By the partial derivative of ${\bC}_{t}(Z(p,\alpha))$ with respect to $\alpha$, we get that given $p$ and $t$, $\alpha^*(p,t)$
is a root of the following function\footnote{All logarithms in this paper are taken according to base 2.}
\begin{align*}
f(p,t)=&(1-p)(1-p^t)\log ((1-\alpha)(1-p^t))\\
&+(2p^t-1-p^{t+1})\log(1-(1-\alpha)(1-p^t) )\\
& +(1-p)h(p^t).
\end{align*}

In Figure~\ref{fig:capacityZp} we present plots of the normalized capacity ${\bC}_{t}(Z(p))$ for $t\in \{-1,1,2,3,4\}$. We also compare between ${\bC}_{t}(Z(p))$ and $1-p$, the maximum normalized capacity for the $BSC(p)$ and $BEC(p)$, which is smaller than ${\bC}_{t}(Z(p))$  for almost all the values of $t$.
Following these computational results, we conjecture that ${\bC}_{t}(Z(p))\le{\bC}_{t+1}(Z(p))$ for all $t \ge 0$ and thus, $F_1(Z(p),D,T)=\min\{D,D_T(Z(p))\}\cdot {\bC}_{T}(Z(p)).$
\begin{figure}[h!]
	\raggedright
\leftskip-1.2em
	\includegraphics[scale=0.52]{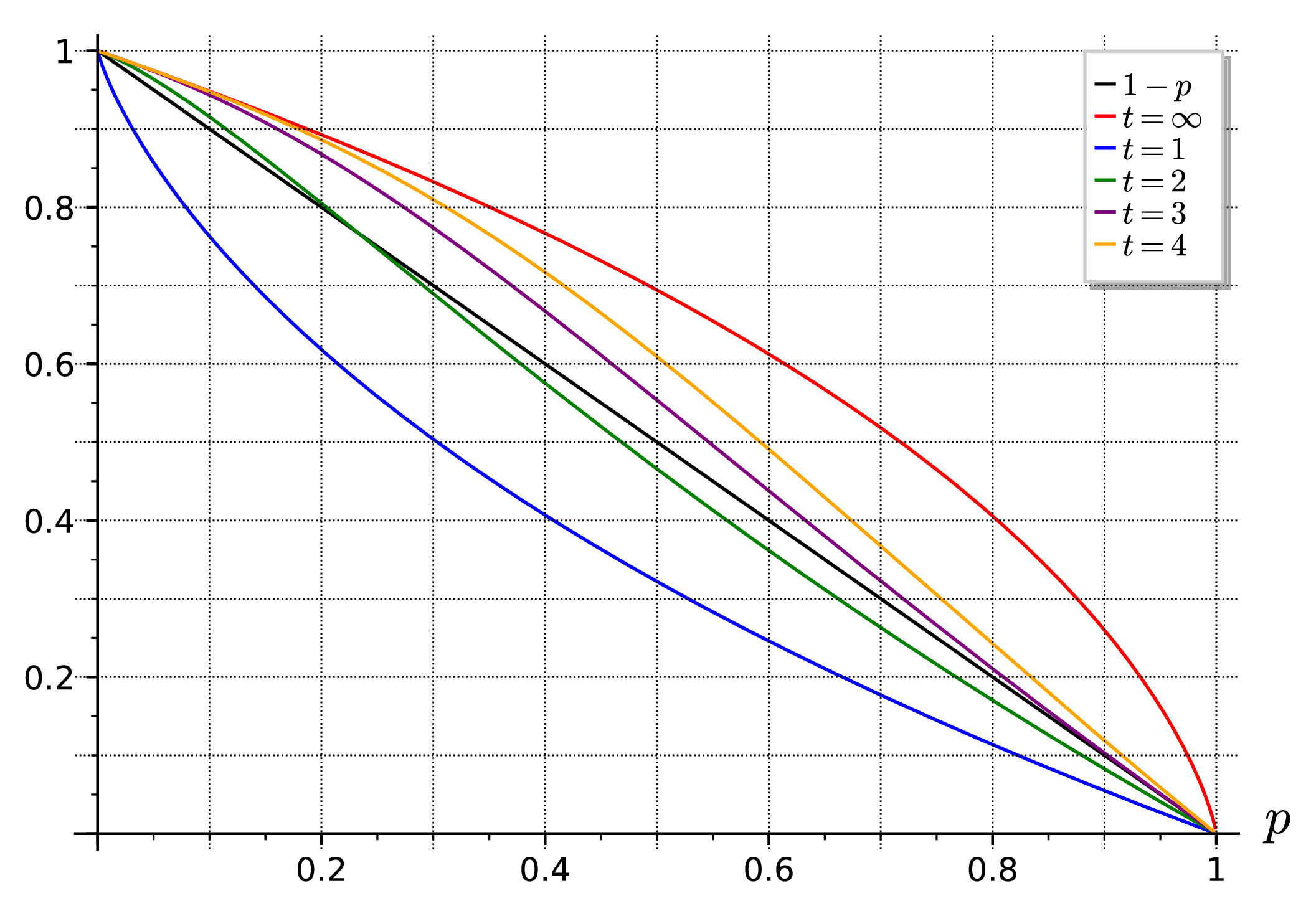}
	\caption{The normalized capacity of $PS_t$ over $Z(p)$ for some values of~$t$,
		comparing to $1-p$, the maximum normalized capacity for $BSC(p)$ and $BEC(p)$.}
	\label{fig:capacityZp}
\end{figure}

There are several similar models which can be solved with the same technique used for the $Z$ channel, for example, the asymmetric programming schemes set up.
Consider the BSC
in which if an error of $0\to 1$ is occurred, we use $PS_{t_1}$, 
and for an error $1\to 0$ we use $PS_{t_2}$.
The problem of finding an optimal programming scheme under this setup, can be solved in the same technique as for the $Z$ channel.

\section{Different Error Probabilities}\label{sec:generalProb1}
In this section we generalize the programming model we studied so far. We no longer assume that there is only a single channel which mimics the cell 
programming attempts, but each programming attempt has its own channel. We study and 
formulate this generalization only for the BSC and the BEC, however modifications for other channels can be handled similarly.

For the rest of this section, we refer the channel $C$ to either $BSC(p)$ or $BEC(p)$. We assume that it is possible to reprogram the cells, however the error probabilities on different programming attempts may be 
different. For example, for \emph{hard} cells in flash memories~\cite{LuHsWaISPP08,SuSuLiISPP95}, i.e., cells that their programming is more difficult, if the first attempt of a cell programming has failed, then the probability for failure on the second trial may be larger since the cell is hard to be programmed. In other cases, the error probability in the next attempt may be smaller since the previous trials might increase the success probability of the subsequent programming attempts.

Let $\bfP=(p_1,p_2,\dots)=(p_i)_{i=1}^{\infty}$ be a probabilities sequence, where $p_t$ is the error probability on the $t$-th programming attempt. We model the programming process as a transmission over the \emph{channel sequence}, $C(\bfP)=C(p_i)_{i=1}^{\infty}$, where on the $t$-th trial, the programming is modeled as transmission over the channel $C(p_t)$. That is, all the channels in $C(\bfP)$ 
have the same type of errors, but may have different error probabilities. Recall that for the BSC we assume that $0\le p_i \le 0.5$ for all $i\ge 1$, while for the BEC, $0\le p_i \le 1$.

For $t\geq -1$ and a channel sequence $C(\bfP)$, we denote by $D_t(C(\bfP))$ 
the average delay of the programming scheme $PS_t$, which is the expected number of 
times to program a cell when the programming process is modeled by $C(\bfP)$. 
For example, for the BSC (see Lemma~\ref{lem:difErrorProb_D_R}), 
\vspace{-2ex}
$$D_t(BSC(\bfP)) =  \sum_{i=0}^{t-1}\left(\prod_{j=1}^{i-1}p_j\right).$$

When a cell is programmed according to a programming scheme $PS_t$, we can model this process as transmission over the set of channels $\{C(p_i)\}_{i=1}^{t}$, and an error occurs if and only if there is an error in each of the $t$ channels. We denote this as a new channel $C_t(\bfP)$, and the capacity of this channel is denoted by $\cC_t(C(\bfP))$. Define $Q_i = \Pi_{j=1}^i p_j$ for $i\ge 1$. Then, for example, for $C=BSC$ we get $BSC_t(\bfP)= BSC(Q_t)$, and the capacity of this channel for $t\ge 1$ is 
$\cC_t(BSC(\bfP))=\cC(BSC(Q_t))=1-h(Q_t)$.

We focus on the set $\cP_T$ of the programming schemes that was defined in~(\ref{eq:programming schemes_T}). It can be readily verified that the average delay of a 
programming scheme $PS\in\cP_T$, $PS = PS\left((\alpha_1, 
t_1),\ldots,(\alpha_{\ell},t_{\ell})\right)$, over the channel sequence 
$C(\bfP)$ is given by
\vspace{-2ex} 
$$D_{PS}(C(\bfP)) =  
\sum_{i=1}^t\alpha_iD_{t_i}(C(\bfP)),$$ and the definition of the capacity is extended
as follows 
\vspace{-2ex}
$$\cC_{PS}(C(\bfP)) = \sum_{i=1}^t\alpha_i \cC_{t_i}(C(\bfP)).$$ 

We are now ready to formally define the problem we study in this section.
\begin{problem}\label{prob:mainVectorProb}
Given a probabilities sequence $\bfP$ with a channel $C\in\{BSC,BEC\}$, an average delay $D$, and a maximum delay $T$, find a programming scheme $PS\in \cP_T$, which maximizes the capacity $\cC_{PS}(C(\bfP))$, under the constraint that $D_{PS}(C(\bfP))\le D$. In particular, find the value  of
$$F_2( C(\bfP), D,T ) = \max_{PS\in \cP_T:D_{PS}(C(\bfP))\le D}\{\cC_{PS}(C(\bfP))\}.$$
\end{problem}

We note that the results presented in Section~\ref{sec:mainProb} regarding Problem~\ref{prob:main} can be derived from the solutions for Problem~\ref{prob:mainVectorProb} presented in this section by substituting $p_i=p$ for all $i\ge 1$.

For $\bfP=(p_1,p_2,\ldots)$ and $Q_i = \Pi_{j=1}^i p_j$, define
$Y_t\deff \sum_{i=1}^{t-1}Q_i$ for $t\ge 1$ ($Y_1=0$), and $Y_{-1}=\deff \sum_{i=1}^{\infty}Q_i$. The next lemma 
establishes the basic properties on the average delay and the capacity of these channels.
	\begin{lemma}\label{lem:difErrorProb_D_R}
	For the programming scheme $PS_t$,
	and $\bfP=(p_1,p_2,\ldots)$,
	the following properties hold:
	\begin{enumerate}
		\item For $t\ge 1$, $ \cC_t(BSC(\bfP))=1-h(Q_t)$,
		\item For $t\ge 1$, $ \cC_t(BEC(\bfP))=1-Q_t$,
		\item $\cC_{-1}(BSC(\bfP))=\cC_{-1}(BEC(\bfP))=1$,
		\item For $t \hspace{-0.07cm}\ge \hspace{-0.07cm}\hspace{-0.05cm} 1$, $D_t(\bfP)  \hspace{-0.07cm}\deff  \hspace{-0.07cm}D_t(BSC(\bfP)) \hspace{-0.07cm}=  \hspace{-0.07cm}D_t(BEC(\bfP)) \hspace{-0.07cm}=  \hspace{-0.07cm}1 \hspace{-0.07cm}+ \hspace{-0.07cm} Y_t$,
		\item $D_0(\bfP)=D_0(BSC(\bfP))= D_0(BEC(\bfP))= 0$,
		\item $D_{-1}(\bfP)\hspace{-0.05cm}=\hspace{-0.05cm}D_{-1}(BSC(\bfP))\hspace{-0.05cm}= \hspace{-0.05cm}D_{-1}(BEC(\bfP))\hspace{-0.05cm}=\hspace{-0.05cm} 1+Y_{-1}$.
	\end{enumerate}
\end{lemma}
\begin{proof}
Note that $Q_t$ is the probability of an error in the first~$t$ attempts. Using 
the known capacities of the BSC and the BEC, we get the values for the 
capacities in cases 1-4.
	
The average delay of the programming scheme $PS_t$ over the channel sequence $BSC(\bfP)$ or $BEC(\bfP)$, which we denoted by $D_t(C(\bfP))$, is calculated as follows. Let $q_i$ be the probability that a cell is programmed at least $i$ times. Note that for $1 <  i < t$, $q_i=Q_{i-1}$ and $q_1=1$ for both cases.
Then, 
we conclude that for $t\ge 1$, 
$$D_t(C(\bfP))=\sum_{i=1}^t q_i =1+ 
\sum_{i=1}^{t-1}Q_i=1+Y_t,$$
and $D_{-1}(C(\bfP))=\sum_{i=1}^\infty q_i =1+ 
\sum_{i=1}^{\infty}Q_i= 1+Y_{-1}$.

\end{proof}

For this generalization of the problem, given a programming scheme $PS\in\cP_T$, the programming scheme $PS(C,D)$ and the normalized capacity are defined in a similar way as in the original definitions in Equations~(\ref{eq:PS(C,D)}) and~(\ref{eq:normalizedCapacity}), respectively.


The following lemma is a generalization of 
Lemma~\ref{lem:connect_prob1AndProb2_1}.
\vspace*{-0.2cm}
\begin{lemma}\label{lem:connect_prob1AndProb2_1_general}
Given a channel sequence $C'=C(\bfP)$, an average delay $D$, and a programming scheme $PS$,
the following holds,
$$\cC_{PS(C',D)}(C')=\min\{D,D_{PS}(C')\}\cdot {\bC}_{PS}(C').$$
\end{lemma}

Next we study the relation between ${\bC}_{PS_t}(C(\bfP))$ and 
${\bC}_{PS_{t+1}}(C(\bfP))$ both for the BSC and the BEC and for arbitrary 
probabilities sequence $\bfP$.

\begin{theorem}\label{thm:BSC-rate_relation}
For $t\ge 1$, and  $\bfP=(p_1, p_2, \ldots)$ such that for all $i$, $0\le p_i  \le 0.5$, there exists
$${\bC}_{PS_t}(BSC(\bfP)) \le {\bC}_{PS_{t+1}}(BSC(\bfP)).$$
\end{theorem}
\begin{proof}
First we state that for all $0\le x,p\le 0.5$ it holds that	
\begin{equation}\label{ineq}
(1-h(x))(1+x)\le 1-h(x/2)\le 1-h(xp).
\end{equation}
Now, we want to prove that
\begin{align*}
\frac{1-h(Q_{t})}{1+Y_{t}} & = {\bC}_{PS_{t}}\left(BSC(\bfP)\right)  \\
\vspace*{-0.2cm}
& \le {\bC}_{PS_{t+1}}\left(BSC(\bfP)\right)=\frac{1-h(Q_{t+1})}{1+Y_{t+1}}.
\end{align*}
This is equivalent to prove that
$$1-h(Q_{t}) \cdot  \frac{1+Y_{t+1}}{1+Y_{t}}\le
1-h(Q_{t+1}),$$
or\vspace*{-0.2cm}
$$1-h(Q_{t}) \cdot \frac{1+Y_{t}+Q_{t}}{1+Y_t} \le
1-h(p_{t+1}Q_{t}),$$
which holds if and only if
\vspace*{-0.2cm}
$$1-h(Q_{t}) \cdot \left( 1+\frac{Q_{t}}{1+Y_t} \right)\le
1-h(p_{t+1}Q_{t}).
\vspace*{-0.2cm}$$
	But, 
	\vspace*{-0.2cm} $$\frac{Q_{t}}{1+Y_t} \le Q_{t}$$ 
	and by substituting $x=Q_t$ and $p=p_{t+1}$ 
	in Inequality~(\ref{ineq}) we conclude that
	\begin{align*}
	(1-h(Q_{t}))\cdot \left(1+\frac{Q_{t}}{1+Y_t}\right)
	 & \le (1-h(Q_{t}))\cdot (1+Q_{t})\\
	 & \le (1-h(p_{t+1}Q_{t})),	
	 \vspace*{-0.2cm}
	\end{align*}
	and therefore 
	$$	{\bC}_{PS_{t}}\left(BSC(\bfP)\right)
	\le {\bC}_{PS_{t+1}}\left(BSC(\bfP)\right),$$
	as required.
\end{proof}

By the previous lemma we conclude the following 
corollary, which its proof is similar to the proof of 
Corollary~\ref{cor:solutionProb1},
using Lemma~\ref{lem:connect_prob1AndProb2_1_general}.
\begin{corollary}\label{or:solutionProb3-BSC}
For a channel sequence $C'=BSC(\bfP)$ the solution for Problem~\ref{prob:mainVectorProb} for $C'$ is
	 $F_2( C', D , T) =\cC_{PS_{T}(C',D)}(C') $ and it is obtained by the programming scheme $PS_{T}(C',D)$. 
\end{corollary}

In the rest of this section, we solve a special case for $BEC(\bfP)$.
\begin{theorem}\label{thm:BEC-rate_relation}
For a probabilities sequence $\bfP=(p_1, p_2, \ldots)$, for all $t\ge 1$, 
$${\bC}_{PS_{t}}\left(BEC(\bfP)\right)\le {\bC}_{PS_{t+1}}\left(BEC(\bfP)\right),$$
if and only if $$p_{t+1}\le \frac{Y_{t+1}}{Y_t + 1}.$$ 
\end{theorem}
\begin{proof}
According to Lemma~\ref{lem:difErrorProb_D_R}, the following relation holds
$$\frac{1-Q_t}{1+Y_{t}}\hspace{-0.3ex}=\hspace{-0.3ex}{\bC}_{PS_{t}}\left(BEC(\bfP)\right) \hspace{-0.3ex}\le\hspace{-0.3ex} {\bC}_{PS_{t+1}}\left(BEC(\bfP)\right)\hspace{-0.3ex}=\hspace{-0.3ex}\frac{1-Q_{t+1}}{1+Y_{t+1}},$$
if and only if
$$ (1-Q_t)\cdot(1+Y_{t+1})\le (1-Q_{t+1}) \cdot (1+Y_{t}).$$
This holds if and only if
$$ -Q_t-Q_tY_{t+1}+Y_{t+1} \le -Q_{t+1}-Q_{t+1}Y_{t}+Y_{t}$$
or
\vspace*{-0.2cm}
$$ Y_{t+1} -Q_t - Y_{t}  -Q_tY_{t+1} \le -Q_{t+1}-Q_{t+1}Y_{t},$$
which translates to
$$ -Q_tY_{t+1} + Q_{t+1} + Q_{t+1}Y_{t} \le 0,$$
and
$$ Q_{t}p_{t+1}(1 + Y_{t}) \le  Q_tY_{t+1},$$
and finally
$$ p_{t+1}\le \frac{Y_{t+1}}{(1 + Y_{t})}.\vspace*{-0.5cm}$$
\end{proof}

\begin{theorem}\label{th:BEC_P}
Let $\bfP=(p_1, p_2, \ldots)$ be a probabilities sequence such that $1 \ge 
p_1\ge p_2 \ge p_3 \cdots$. Then, for all $t\ge 1$,
$${\bC}_{PS_{t}}\left(BEC(\bfP)\right)\le {\bC}_{PS_{t+1}}\left(BEC(\bfP)\right).$$
\end{theorem}
\begin{proof}
According to Theorem~\ref{thm:BEC-rate_relation}
$${\bC}_{PS_{t}}\left(BEC(\bfP)\right)\le {\bC}_{PS_{t+1}}\left(BEC(\bfP)\right)$$
if and only if
\vspace*{-0.2cm}
$$ p_{t+1} \le \frac{Y_{t+1}}{(1 + Y_{t})}
\vspace*{-0.2cm}$$
or	
$$ p_{t+1}\left(1 + Y_{t}\right) \le Y_{t+1}$$
and by the definition of $Y_t$
$$ p_{t+1} + p_{t+1}\left(\sum_{i=1}^{t-1}Q_i\right) \le \sum_{i=1}^{t}Q_i
\vspace*{-0.2cm}$$
and thus	
\vspace*{-0.2cm}
$$ p_{t+1} + p_{t+1}\left(\sum_{i=1}^{t-1}Q_i\right)- \sum_{i=1}^{t}Q_i \le 0.$$
By $Q_i$ definition and since $p_1\ge p_2 \ge p_3 \cdots$ we have 
	$$p_{t+1}\left(\sum_{i=1}^{t-1}Q_i\right) \le p_t 
	\left(\sum_{i=1}^{t-1}Q_i\right) 
	\leq \sum_{i=2}^{t}Q_i,$$
	and by $p_{t+1}\le p_1=Q_1$
	we conclude that
	$$ p_{t+1} + p_{t+1}\left(\sum_{i=1}^{t-1}Q_i\right)- \sum_{i=1}^{t}Q_i \le 
	0,$$
	\vspace*{-0.1cm}
	as required.
\end{proof}

By Theorem~\ref{th:BEC_P} we can finally conclude with the following corollary.
\begin{corollary}\label{cor:solutionProb3-BEC}
For a channel sequence $C'=BEC(\bfP)$ where
$\bfP=(p_1, p_2, \ldots)$ such that
$1 \ge p_1\ge p_2 \ge p_3 \cdots $,
the solution for Problem~\ref{prob:mainVectorProb} for $C'$ is
$F_2( C', D , T) =\cC_{PS_{T}(C',D)}(C') $ and it is obtained by the programming scheme $PS_{T}(C',D)$.
\end{corollary}

\section{Combined Programming Schemes for the BSC and the BEC}\label{sec:psAsBSCandBEC}
In this section, we study programming schemes for the BSC, in which on the last programming attempt it is possible to either try to reprogram the failed cell again with its value or instead program it with a special question mark to indicate a programming failure. This model is motivated by several applications. For example, when synthesizing DNA strands, if the attachment of the next base to the strand fails on multiple attempts, it is possible to attach instead a different molecule to indicate this base attachment failure~\cite{KC14}. In flash memories we assume that if some cell cannot reach its correct value, then it will be possible to program it to a different level (for example a high voltage level that is usually not used) in order to indicate a programming failure of the cell.

We denote by $PS_{q,t}$ the programming scheme in which on the $t$-th programming attempt, which is the last one, the cell is programed without verification with probability $q$, and with probability $1-q$ it is programmed with the question mark symbol $'?'$.

The average delay of programming a cell with $PS_{q,t}$ over the $BSC(p)$ does not depend on $q$, and hence equals to $D_t(p)$.
However the capacity is clearly influenced by the parameter~$q$.

Let $p$ be the programming error probability. Then, the probability that a cell will be erroneous after $t-1$ programming attempts is $p^{t-1}$.
Therefore, programming with $PS_{q,t}$ over $BSC(p)$ can be represented by a channel 
with the following transitions probabilities
$$
p(y|x)=
\begin{cases}
p^{t-1}(1-q)&\mbox{if } y=?\\ 
p^{t-1}q(1-p)+(1-p^{t-1})& \mbox{if } x=y\\ 
p^{t-1}qp& \mbox{otherwise} ,\\ 
\end{cases}
$$
where $x$, $y$ is the input, output bit of the channel, respectively. 
Denote $b=p^{t-1}$. The capacity of this channel is \cite[Porblem~7.13]{CT91}
\begin{eqnarray*}
\cC_{q,t}(BSC(p))
&\hspace*{-1.5ex}=
&\hspace*{-1.5ex}(1-b+bq)\left(1-h\left(\frac{bpq}{1-b+bq}\right)\right) \\
&\hspace*{-1.5ex}=&\hspace*{-1.5ex}1-b+bq\\
&\hspace*{-1.5ex}&\hspace*{-1.5ex}-(1-b+bq)\log(1-b+bq)\\
&\hspace*{-1.5ex}&\hspace*{-1.5ex}+(1-b+bq-bqp)\log(1-b+bq-bqp)\\
&\hspace*{-1.5ex}&\hspace*{-1.5ex}+bpq\log(bpq).
\end{eqnarray*}

Note that,
$\cC_{0,t}(BSC(p))=1-p^{t-1}$,
and
$\cC_{1,t}(BSC(p))=1-h(p^t)$.
For example, for $t=1$,
$$
\cC_{q,1}(BSC(p)) \hspace{-0.1cm}= \hspace{-0.1cm}q-q\log(q) + q(1-p)\log(q(1-p))
+qp\log(qp).
$$

Let $PS=PS\left((\beta_1, t_1),(\beta_2, t_2),\ldots, (\beta_{\ell},t_{\ell})\right)\in \cP_T$ be a programming scheme, and $\bfq=(q_1,q_2,\ldots,q_{\ell})$
where $0\le q_i \le 1$, for all $1\leq i\leq \ell$.
Then, we define 
$$\cC_{PS,\bfq}(BSC(p))=\sum_{i=1}^{\ell} 
\beta_i\cdot \cC_{q_i,t_i}(BSC(p)).$$
That is, $\cC_{PS,\bfq}(BSC(p))$ is the capacity of $BSC(p)$ when using the programming scheme $PS$ with the parameter $\bfq$.
Similarly, we define the average delay of the programming scheme $PS$ for $BSC(p)$ using the parameter $\bfq$ as
$$D_{PS,\bfq}(BSC(p))=
\sum_{i=1}^{\ell} 
\beta_i\cdot \cD_{q_i,t_i}(BSC(p)).$$
Note that $D_{PS,\bfq}(BSC(p))=D_{PS}(BSC(p)).$

For this model, Problem~\ref{prob:main} will be formulated as follows.
\begin{customProblem}{\ref{prob:main} - Combined channel} \label{prob:main_combined}
	Given a channel $BSC(p)$, an average delay $D$, and a maximum delay $T$, find a programming scheme, $PS\in \cP_T$, and $\bfq$ which maximize the capacity $\cC_{PS,\bfq}(BSC(p))$, under the constraint that $D_{PS,\bfq}(BSC(p))\le D$. 
	In particular, given $BSC(p)$, $D$, and $T$, find the value of $$F_3(BSC(p),D,T) = \max_{D_{PS,\bfq}(BSC(p))\le D}\{\cC_{PS,\bfq}(BSC(p))\}.$$
\end{customProblem}

For this generalization of the model, the programming scheme $PS(C,D)$ and the normalized capacity are defined in a similar way as in the original definitions in Equations~(\ref{eq:PS(C,D)}) and~(\ref{eq:normalizedCapacity}), respectively.

Given $p,t$ we define $$\cC'_{t}(BSC(p))=\max_{q\in[0,1]}\{\cC_{q,t}(BSC(p))\},$$
and the normalized capacity
${\bC'}_{t}(BSC(p))=\frac{\cC'_{t}(BSC(p))}{D_{t}(p)}$.

In the rest of this section we prove that the best scheme is  $PS_{1,T}(C,D)$ or $PS_{0,T}(C,D)$, i.e., the standard $PS_T$
or the new $PS_T$ in which in the last attempt all the erroneous cells are programmed with a question mark. 

\begin{lemma}\label{lem:combinedMaxq}
	Given $p$ and $t$,  $$\cC'_{t}(BSC(p))=\max\{\cC_{0,t}(BSC(p)),\cC_{1,t}(BSC(p))\}.$$
\end{lemma}
\begin{proof}
	If $p=0$ then there is not errors, and the maximum capacity is obtained for all $q$.
	Given $0<p\le 0.5$,	if $t=1$ then $\cC_{q,t}(BSC(p))=q(1-h(p))$ and the maximum is obtained for $q=1$.
	Otherwise, given $p,t$, such that $0<p\le 0.5$ and $1<t$, we prove that the function $\cC_{q,t}(BSC(p))$ has no local maximum in the range of $0< q<1$. 
	We define $f_{t,p}(q)=\cC_{q,t}(BSC(p))$ as a function of $q$,
	and prove that $f_{t,p}(q)$ has no local maximum in the range of $0< q<1$ by showing that the second derivation of $f_{t,p}(q)=\cC_{q,t}(BSC(p))$ is positive in that range.

	The first derivation is
	\begin{align*}
	\frac{\partial f_{t,p}(q)}{\partial q}= &
	b(1-p)\log(1-b+bq-bqp)\\
	& + bp\log(bpq) - b\log(1-b+bq	)+ b,
	\end{align*}
	and then the second derivation is
	\begin{align*}
	\frac{\partial f_{t,p}(q)}{\partial q^2}= &
		\frac{b^2(1-p)^2}{(1-b+bq-bpq)\ln 2}\\
	& +\frac{(bp)^2}{bpq \ln 2}\\
	& -\frac{b^2}{(1-b+bq)\ln 2}.
	\end{align*}
	To show that $\frac{\partial f_{t,p}(q)}{\partial q^2}>0$, is sufficient to prove that 
	 $$	\frac{(1-p)^2}{1-b+bq-bpq}+
	 \frac{p^2}{bpq}-
	 \frac{1}{1-b+bq}>0.$$
	 We denote $x_1=(1-b+bq-bpq)$ and $x_2=bpq$ 
	 Thus, we want to prove that  
	 $$\frac{(1-p)^2}{x_1}+
	 \frac{p^2}{x_2}-
	 \frac{1}{x_1+x_2}>0.$$
	Note that $x_1=1-b(1-q(1-p))>0$ and $x_2=bpq>0$ since $0<b,p\le 0.5$ and $0<q<1$.
	Thus, we can prove that
	 $$(1-p)^2x_2(x_1+x_2)+p^2x_1(x_1+x_2)-x_1x_2>0,$$
	 which hold if
	  $$((1-p)x_2-px_1)^2>0.$$
	  But the last equation holds since 
	  $(1-p)x_2= px_1$ implies 
	  $0=p(1-b)$ which is impossible since $0<p\le 1/2$, $1<t$, and $b=p^{t-1}$.
\end{proof}
The last lemma proved that for all $p,t$, the capacity $\cC_{q,t}(BSC(p))$ is achieved for $q=0$ or
for $q=1$, by comparing between $\cC_{0,t}(BSC(p))=1-p^{t-1}$
which obtained for $PS_{0,t}\ (q=0)$
and $\cC_{1,t}(BSC(p))=1-h(p^t)$ which attained for $PS_{1,t}\ (q=1)$.
This result can be intuitively explained as if $q=0$ then the last programming is just providing a complete verification for all the successful cells by substituting a question mark in all the erroneous cells.
Thus, according to some threshold ($p,t$), 
we can either provide a complete verification
for the already programmed cells ($q=0$) 
or try to reprogram again all the failed cells ($q=1$). 

\begin{theorem}\label{thm:errorProbP_NORM-R-Combined}
	For $t\ge 0$,
	$${\bC'}_{t}(BSC(p)) \le {\bC'}_{t+1}(BSC(p))$$
\end{theorem}
\begin{proof}
	By Lemma~\ref{lem:combinedMaxq}
	$\cC'_{t}(BSC(p))=\max\{\cC_{0,t}(BSC(p)),\cC_{1,t}(BSC(p))\}$.
	If $\cC'_{t}(BSC(p))$ is obtained for $q=1$, 
	then the claim is implied by Theorem~\ref{thm:errorProbP_NORM-R}.
  Otherwise, 	$\bC'_{t}(BSC(p))=\frac{(1-p)(1-p^{t-1})}{1-p^t}$
  and
  $\bC'_{t+1}(BSC(p))\ge \frac{(1-p)(1-p^{t})}{1-p^{t+1}}.$
  Then, the claim is true since $p\le 1/2$ implies
  $\frac{1-p^{t}}{1-p^{t+1}} \ge \frac{1-p^{t-1}}{1-p^t}$.
\end{proof}

By applying the same technique as in the proof of Corollary~\ref{cor:solutionProb1} with using Lemma~\ref{lem:combinedMaxq} and Theorem~\ref{thm:errorProbP_NORM-R-Combined} we solve Problem~\ref{prob:main} for the new model.
\begin{corollary}\label{cor:solutionProb1_combined}
	
	Denote by $D'=\min\{ D_T(p),D\}$ then the solution for Problem~\ref{prob:main_combined} is
	$F_3(BSC(p), D,T) =\frac{D'}{D_T(p)}\cdot \max\{1-h(p^t),1-p^{t-1}\}$
	obtained by $PS_{1,T}(D)$ or $PS_{0,T}(D)$, respectively.	
\end{corollary}

For each $p$ we denote by $t_p$ the smallest value of $t$, such that $h(p^t)\ge p^{t-1}$
($t_p$ my be non integer).
Since $h(px)\ge ph(x)$ for $0 \le x,p\le 0.5$, we conclude that for each $t\ge t_p$ there exists
$h(p^t)\ge p^{t-1}$.
Thus, given $p$, there exists $t\ge t_p$ if and only if
$\cC_{0,t}(BSC(p))=1-p^{t-1} \ge 1-h(p^t)=\cC_{1,t}(BSC(p))$.
Let $T$ be the last attempt to program. If $T\ge t_p$ then the encoder in the $T$-th attempt will program question marks in all the failed cells.
Otherwise, in the $T$-th attempt, the failed cells will be programmed (without a verification). In Figure~\ref{fig:tpGraph} the $t_p$ values are presented in a graph, where the horizontal axis is $p$ and the vertical axis is $t$. The graph line is $f(p)=t_p$.

\begin{figure}[h!]
	\includegraphics[scale=0.6]{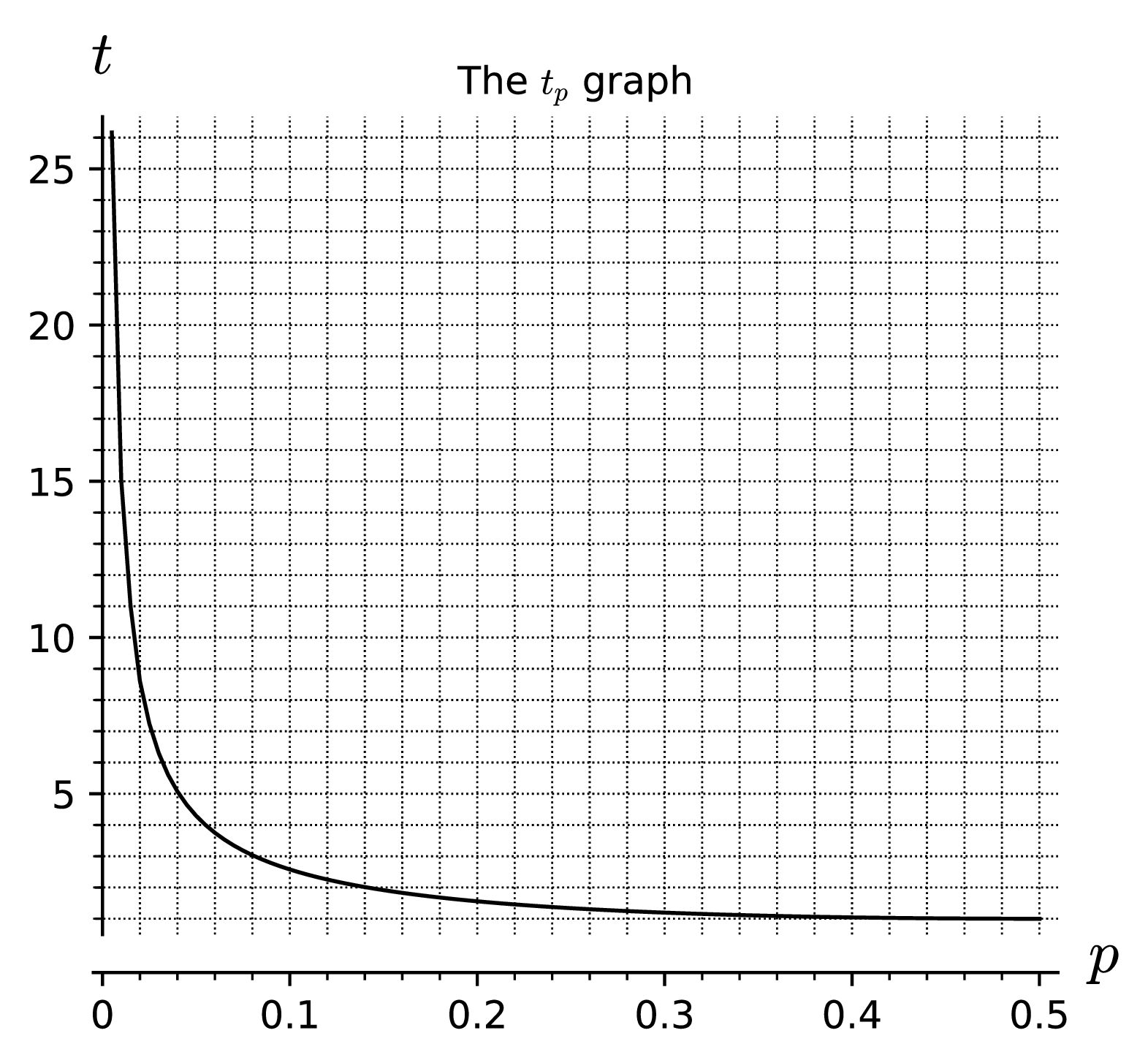}
	\caption{The $t_p$ graph.}
	\label{fig:tpGraph}
\end{figure}

\section{Conclusion}\label{sec:conc}
In this paper we studied a model which described the process of cell programming in memories. We focused on the case where the programming is modeled by the BSC, the BEC and the $Z$ channel, and accordingly, we designed programming schemes that maximize the number of information bits that can be reliably stored in the memory, while the average and maximum numbers of times to program a cell are constrained. While this work established several interesting observations on the programming strategies in memories and transmission schemes, there are still several questions that remain open. In particular, the generalization of this model to multilevel cells, and to a setup in which the cells are programmed in parallel.

\appendices
\section{}\label{appA}
In this part we present the omitted proofs in the paper.

\begin{customCor}{\ref{cor:solutionProb1Z-channel}}
	$F_1(Z(p), D,T) =\min\{D_T(Z(p)),D\}\cdot{\bC}_{PS_{t^*(D)}}\left(Z(p)\right)=
	{\cC}_{PS_{t^*(D)}(Z(p),D)}\left(Z(p)\right)$
	obtained by $PS_{t^*(T)}(Z(p),D)$ with parameter $\alpha^*(p,t^*(T))$.
\end{customCor}
\begin{proof}
	Let $PS = PS\left((\beta_1, t_1),\ldots, (\beta_{\ell},t_{\ell})\right) \in \cP_T$ be
	a programming scheme which meets the constraint $D$ with the parameter ${\balpha}=(\alpha_1,\ldots, \alpha_t)$.
	Thus, we have
	\begin{align*}
	\cC_{PS}(Z(p,\balpha))
	& =\sum_{i=1}^{\ell} 
	\beta_i\cdot \cC_{t_i}(Z(p,\alpha_i))\\
	& \underset{(1)}{=}\sum_{i=1}^{\ell} 
	\beta_i\cdot D_{t_i}(Z(p,\alpha_i)) \cdot {\bC}_{PS_{t_i}}(Z(p,\alpha_i))\\
	& \underset{(2)}{\le}\sum_{i=1}^{\ell} 
	\beta_i\cdot D_{t_i}(Z(p,\alpha_i)) \cdot {\bC}_{PS_{t^*(T)}}(Z(p))\\
	& = {\bC}_{PS_{t^*(T)}}(Z(p)) \sum_{i=1}^{\ell} 
	\beta_i\cdot D_{t_i}(Z(p,\alpha_i))\\
	& \underset{(3)}{\le} {\bC}_{PS_{t^*(T)}}(Z(p)) \cdot D
	\end{align*}
	where $(1)$ is by the definition of the normalized capacity,
	$(2)$ is by $t^*(T)$ and ${\bC}_{PS_{t^*(T)}}(Z(p))$ definitions,
	and
	$(3)$ is since $PS$ meets the average delay constraint $D$ with parameter ${\balpha}$. 
\end{proof}


\begin{thebibliography}{10} 

\bibitem{Betal16}
M. Blawat, K. Gaedke, I. H\"utter, X.-M. Chen,  B. Turczyk, S. Inverso, B.W. Pruitt, and G.M. Church, ``Forward error correction for DNA data storage," \emph{Int. Conf. on Computational Science}, vol. 80, pp. 1011--1022, 2016.

\bibitem{B16}
J. Bornholt, R. Lopez, D.M. Carmean, L. Ceze, G. Seelig, and K. Strauss, ``A DNA-based archival storage system," \emph{Proc. of the Twenty-First Int. Conf. on Architectural Support for Programming Languages and Operating Systems (ASPLOS)}, pp. 637--649, Atlanta, GA, Apr. 2016.

\bibitem{BL10}
C. Bunte and A. Lapidoth, ``On the storage capacity of rewritable memories,'' \emph{IEEE 26-th Convention of Electrical and Electronics Engineers in Israel}, pp. 402--405, 2010.

\bibitem{BL11}
C. Bunte and A. Lapidoth, ``Computing the capacity of rewritable memories,'' \emph{Proc. Int. Symp. on Inform. Theory}, pp. 2512--2516, St. Petersburg, Jul. 2011.

\bibitem{CT91}
{T.M.\,Cover and J.A.\,Thomas}, ``Elements of Information Theory,'' 1st Edition. New York: Wiley-Interscience, 1991.

\bibitem{EZ17}
Y. Erlich and D. Zielinski, ``DNA fountain enables a robust and efficient storage architecture," \emph{Science}, vol. 355, no. 6328, pp. 950--954, 2017.

\bibitem{FLMS09} 
M. Franceschini, L. Lastras-Montano, T. Mittelholzer, and M. Sharma, ``The role of feedback in rewritable storage channels [Lecture Notes],''
\emph{IEEE Signal Processing Magazine}, vol. 26, pp. 190--194, 2009.

\bibitem{Getal13}
N. Goldman, P. Bertone, S. Chen, C. Dessimoz, E.M. LeProust, B. Sipos, and E. Birney, ``Towards practical, high-capacity, low-maintenance information storage in synthesized DNA," \emph{Nature}, vol. 494, no. 7435, pp. 77--80, 2013.
    
\bibitem{JiBrISITA08}
A. Jiang and J. Bruck, ``On the capacity of flash memories,'' \emph{Proc. Int. Symp. on Inform. Theory and Its Applications}, pp. 94--99, 2008.
    
\bibitem{JiLiPACRIM09}
A. Jiang and H. Li, ``Optimized cell programming for flash memories,'' \emph{Proc. IEEE Pacific Rim Conference on Communications, Computers and Signal Processing (PACRIM)}, pp. 914--919, 2009.

\bibitem{KC14}
S. Kosuri and G.M. Church, ``Large-scale de novo DNA synthesis: Technologies and applications," \emph{Nature Methods}, vol. 11, no.5, pp. 499--507, May 2014. 

\bibitem{LFM10}
L. Lastras-Montano, M. Franceschini, and T. Mittelholzer, ``The capacity of the uniform noise rewritable channel with average cost,'' \emph{Proc. IEEE Int. Symp. of Inform. Theory}, pp. 201--205, Austin, TX, Jun. 2010.

\bibitem{LFMS08}
L. Lastras-Montano, M. Franceschini, T. Mittelholzer, and M. Sharma, ``Rewritable storage channels,'' \emph{Proc. IEEE Int. Symp. of Inform. Theory}, pp. 7--10, Toronto, Jul. 2008.

\bibitem{LFMS14}
L. Lastras-Montano, M. Franceschini, T. Mittelholzer, and M. Sharma, ``On the Capacity of Memoryless Rewritable Storage Channels,'' 
\emph{IEEE Trans.\ Inform.\ Theory,} vol. 60, no. 6, pp. 3178--3195, Jun. 2014.


\bibitem{LMF10}
L. Lastras-Montano, T. Mittelholzer, and M. Franceschini, ``Superposition coding in rewritable channels,'' \emph{Proc. on Inform. Theory and Applications Workshop}, San Diego, CA, USA Feb. 2010. 

\bibitem{LuHsWaISPP08}
H. T. Lue, T. H. Hsu, S. Y. Wang, E. K. Lai, K. Y. Hsieh, R. Liu, and C. Y. Lu, ``Study of incremental step pulse programming (ISPP) and STI edge effect of BE-SONOS NAND flash,'' \emph{Proc. IEEE Int. Symp. on Reliability Physics}, vol. 30, no. 11, pp. 693--694, May 2008.

\bibitem{MFLES09} 
T. Mittelholzer, M. Franceschini, L. Lastras-Montano, I. Elfadel, and M. Sharma, ``Rewritable channels with data-dependent noise,'' 
\emph{International Conference on Communications}, pp. 1--6, 2009.

\bibitem{MLSF10}
T. Mittelholzer, L. Lastras-Montano, M. Sharma, and M. Franceschini, ``Rewritable storage channels with limited number of rewrite iterations,''
 \emph{Proc. IEEE Int. Symp. of Inform. Theory}, pp. 973--977, Austin, TX, Jun. 2010.

\bibitem{QJS13}
M. Qin, A. Jiang, and P.H. Siegel, ``Parallel programming of rank modulation," \emph{Proc. IEEE Int. Symp. of Inform. Theory}, pp. 719--723, Istanbul, Turkey, Jul. 2013.

\bibitem{QYS14}
M. Qin, E. Yaakobi, and P.H. Siegel, ``Optimized cell programming for flash memories with quantizers," \emph{IEEE Trans.\ Inform.\ Theory,} vol. 60, no. 5, pp. 2780--2795, May 2014.

\bibitem{SuSuLiISPP95}
K. D. Suh \emph{et al.}, ``A 3.3 V 32 Mb NAND flash memory with incremental step pulse programming,'' \emph{IEEE Journal of Solid-State Circuits}, vol. 30, no. 11, pp. 1149--1156, Nov. 1995.
 
\bibitem{TAB02}
{L.G.\,Tallini, S.\,Al-Bassam, and B.\,Bose}, ``On the capacity and codes for the Z-channel,'' \emph{Proc. IEEE Int. Symp. on Inform. Theory}, pp.\,422, Lausanne, Switzerland, 2002.

\bibitem{VTLF14}
R. Venkataramanan, S. Tatikonda, L. Lastras-Montano, and M. Franceschini, ``Rewritable storage channels with hidden state,'' \emph{IEEE Journal on Selected Areas in Communication}, Vol. 32, No. 5, pp. 815--824, May 2014.

\bibitem{V97}
{S.\,Verdu}, ``Channel capacity,'' Ch. 73.5 in the \emph{Electrical Engineering Handbook, IEEE and CRC Press}, pp.\,1671--1678, 1997.
    
\bibitem{YJSVW10}
{E.\,Yaakobi, A.\,Jiang, P.H.\,Siegel, A.\,Vardy, and J.K.\,Wolf}, ``On the parallel programming of flash memory cells,'' \emph{Proc. IEEE Inf. Theory Workshop}, Dublin, Ireland, Aug.-Sep. 2010.   

\bibitem{YGM16}
S.H.T. Yazdi, R. Gabrys, and O. Milenkovic, ``Portable and error-free DNA-based data storage," \emph{Cold Spring Harbor Labs Journals}, 2016.
\end{thebibliography}
\end{document}